\documentclass[final,5p,times,twocolumn]{elsarticle}

\usepackage{color}
\usepackage{graphicx}
\usepackage{array}
\usepackage{amsmath}
\usepackage[utf8]{inputenc}
\usepackage[english]{babel}
\usepackage{amsthm}
\usepackage{floatrow}
\usepackage{tabularx}
\usepackage{seqsplit}
\theoremstyle{plain}
\newtheorem{proposition}{Proposition}

\usepackage{url}

\usepackage{breakurl}
\usepackage[breaklinks]{hyperref}

\newcolumntype{C}[1]{>{\centering\let\newline\\\arraybackslash\hspace{0pt}}m{#1}}

\usepackage{diagbox}

\usepackage{amsfonts}

\journal{Journal of Parallel and Distributed Computing}

\begin{document}

\begin{frontmatter}

\title{BSSSQS: A Blockchain Based Smart and Secured Scheme for Question Sharing in the Smart Education System}


\author[mymainaddress]{Anik Islam}
\author[mysecondaryaddress]{Md. Fazlul Kader}
\author[mymainaddress]{Soo Young Shin\corref{mycorrespondingauthor}}
\cortext[mycorrespondingauthor]{Corresponding author}
\ead{wdragon@kumoh.ac.kr}



\address[mymainaddress]{Department of IT Convergence Engineering, Kumoh National Institute of Technology, Gumi, 39177, South Korea}
\address[mysecondaryaddress]{Department of Electrical and Electronic Engineering, University of Chittagong, Chittagong, 4331, Bangladesh}

\begin{abstract}
Existing education systems are facing a threat of question paper leaking (QPL) in the exam which jeopardizes the quality of education. Therefore, it is high time to think about a more secure and flexible question sharing system which can prevent QPL issue in the future education system. Blockchain enables a way of creating and storing transactions, contracts or anything that requires protection against tampering, accessing etc. This paper presents a new scheme for smart education, by utilizing the concept of blockchain, for question sharing. A two-phase encryption technique for encrypting question paper (QSP) is proposed. In the first phase, QSPs are encrypted using timestamp and in the second phase, previous encrypted QSPs are encrypted again using a timestamp, salt hash and hashes from previous QSPs. These encrypted QSPs are stored in the blockchain along with a smart contract which helps the user to unlock the selected QSP. An algorithm is also proposed for selecting a QSP for the exam which picks a QSP randomly. Moreover, a timestamp based lock is imposed on the scheme so that no one can decrypt the QSP before the allowed time. Finally, security is analyzed by proving different propositions and the superiority of the proposed scheme over existing schemes is proven through a comparative study based on the different features.
\end{abstract}


\end{frontmatter}


\section{Introduction}
Blockchain is a brilliant discovery which has brought a revolution in the realm of technology. This revolutionary idea was inaugurated by Satoshi \\Nakamoto~\cite{1}, a person or group used this pseudonym in order to keep their individuality concealed. Blockchain is a peer-to-peer (P2P) distributed ledger, introduced within the concept of \textit{bitcoin} cryptocurrencies, which stores the history of transactions with the payload~\cite{2}. However, blockchain technology is not bound within doing financial transactions. It has started to draw the interest of the stakeholders of a wide span of industries which covers finance, healthcare, utilities, real estate, government sector, and digital content distribution etc. ~\cite{3, 4, 5, 6, 8284127,8405627}, due to its flexibility and security mechanism. The concept of blockchain was employed in the database named "BigchainDB", a database which is capable of piling large quantity data~\cite{7}. 
\par Blockchain is a data structure which is distributed and replicated amongst the participated nodes in the network. When a transaction happens in the network, that transaction has to experience validation called consensus mechanism, a process where some of the participants reach a mutual agreement in allowing that transaction ~\cite{8}. Those who perform consensus mechanism, are called miners and they have to execute a computationally hard puzzle. After that, a block is added to the chain including that transaction. Each block contains the hash of the previous block. That is why it is called blockchain. The first block of the blockchain is called genesis block. In the place of previous block's hash, it contains 0 or some other value which means it does not refer the previous block's hash. In most of the blockchain application, genesis block is hard coded~\cite{9}. However, blockchain is tamper proof. If anyone makes any changes in the transaction then the hash of that block also changes which breaks the chain. Then a user has to mine the chain again in order to make the changes valid. Moreover, that user has to make changes in other participants node which makes it very difficult. 
\par Furthermore, asymmetric cryptography is adopted in order to issue transactions in blockchains~\cite{8284128}. In blockchain, every user has two keys, such as 1) public key, and 2) secret key. The public key plays a role as an address in the system. On the other hand, the secret key is utilized to sign transactions~\cite{8425613}. Thus blockchain creates a trustless network where all users’ identities are hidden and parties can transact securely without trusting each other~\cite{2}. 

\par Internet of Things (IoT) has brought another revolution in the realm of technology \cite{10, 11, 12, 8378971}. Recently, IoT has put its mark in the education sector~\cite{13, 14}. Smart campus, smart classroom, digital content, smart exam, remote learning, campus safety etc. are the results of IoT. In order to improve the quality of education, a lot of researches are going on. However, IoT technology is facing security risks. Though there are a lot of existing security protocols, they are not adequate to provide security because of the decentralized structure of IoT. Entities in IoT need reliable and tamper-proof protection from attacks like denial-of-sleep and denial-of-service~\cite{15}. Blockchain can mitigate this issues with its security infrastructure~\cite{16, 17}.
\par Examination is one of the important parts of the education system because it not only evaluates the understanding of students but also forces them to study~\cite{18, 19}. However, there is a threat, named Question Paper Leaking (QPL), which can cause the fairness issues in the examinations. Nowadays, QPL is a serious issue throughout the world from university entrance examination to public examination, and the situation is worse in developing countries~\cite{20, 21, 22, 23}. The QPL can bring some serious outcome, such as (1) quality of education compromised and (2) erosion of ethical standards~\cite{23}.

\par ACT Inc, who is the creator of the United States’ most popular college entrance examination, canceled some college entrance exams after leaking the test materials~\cite{24}. In the United Kingdom, Brighton Hove and Sussex Sixth Form College canceled A-level physics exam after noticing the question paper leak on social media~\cite{25}. In China, a teacher was accused of leaking math test paper of the annual postgraduate entrance exam~\cite{26}. In the University of KwaZulu-Natal School of Applied Human Sciences, at least four exam papers in two subjects were leaked~\cite{27}. In Egypt, French language exam papers were posted after half an hour after the exam started~\cite{28}, and a version of Arabic exam was leaked on the first day of Thanaweya Amma examinations~\cite{29}. A teacher from Vietnam leaked final examination paper to the son of a neighbor~\cite{30}.
In Nepal, the question papers of International English Language Testing System examination~\cite{31} and Bachelor of Medicine and Bachelor of Surgery entrance examination of the Tribhuvan University Institute of Medicine~\cite{32} were leaked. In 2017, a number of incidents related to exam paper leak happened in Pakistan~\cite{33, 34, 35, 36, 37} and India~\cite{38, 39, 40}, respectively. In Bangladesh, Junior School Certificate and Secondary School Certificate examination questions were leaked~\cite{41, 42}. In Korea, a high school teacher was accused of leaking English test questions~\cite{43}. 

\par Though the aforementioned cases~\cite{24, 25, 26, 27, 28, 29, 30, 31, 32, 33, 34, 35, 36, 37, 38, 39, 40, 41, 42, 43}  only covers the QPL incidents happened in 2017, some countries face this problem almost in every year. Hence, it can be said that QPL happens not only within the developing and underdeveloped countries, but also in developed countries. In QPL incidents, not only the students, but also the teachers and authorities are involved. Therefore, it is required to develop a smart examination system which can share examination papers securely without the fear of QPL.

In order to digitalize the examination system, different ideas are shared in the literature~\cite{44, 45, 46, 47, 48, 49, 50}. In~\cite{44}, three models of web examination system, such as B/S, C/S, and B/S combined with C/S are discussed, along with the technical details and solution of the problem of concurrent data. Another web based examination system was proposed for distant and formal education in \cite{45}. In this system, teachers enter examination questions and number of the students. The system automatically shuffles questions and answers so that every student can get a different pattern of questions to prevent cheating. In \cite{46}, an online examination system for PE theory courses was proposed, where every user uses username and password given by the administrators to login the system. The system maintains the question distribution and the time schedule of examination. 
In \cite{47}, an online examination system was proposed, where MD5 encryption technique was exploited for protecting password and resources. Morevoer, a WEB-INF Directory is used so that students cannot access the contents directly without authorization. An examination management system, based on flat network, was demonstrated by \cite{48}. The system provides role based security, where students can only sign up and give examinations. In \cite{49}, a web-based examination system was proposed to integrate with existing learning management systems, whereas an online examination system based on TCP/IP client-server architecture and spiral model was proposed in~\cite{50}. These systems mainly focus on system design and overall management. However, they could not guarantee to solve the QPL incidents. 

\par From the recent incidents of QPL, it is obvious that all levels of people are involved. Moreover, social engineering, phishing etc., can loot anyone credential and can access data anytime. Furthermore, a system with large scale, like implementing a system for the national education system  which demands distributed computing, examination paper sharing and storing, are very vulnerable. Therefore, examination management systems need more than user credential and random question selection. Question sharing (QS) should be performed through a more robust system, where user credential will be less important. In such a system, even though unauthorized persons get the password, they could not be able to access the questions before exams. Blockchain prevents data from altering once it's being mined in the chain. Moreover, blockchain assists to maintain access permission and audit process seamlessly. Blockchain can be one of the promising techniques to provide security against the aforementioned threats.

In this paper, a blockchain based smart and secured QS scheme for smart education system (termed as BSSSQS) is proposed, a topic which has not been explored yet to the best of our knowledge. The major contributions of this paper are compiled as follows.
\begin{itemize}
\item The proposed scheme can increase the security of questions and provide seamless sharing among the examination centers.
\item A two-phase encryption technique (using different parameters as a key) is proposed in order to provide security over question. 
\item A randomization algorithm is proposed for selecting a question paper (QSP) before the exam.
\item A design of smart contract is proposed for managing user authorization over blocks along with decrypting QSP from blockchain. 
\item A timestamp based lock is proposed. In the proposed scheme, every exam center will hold the QSP, but no one can access (or request) it without system permission.   
\end{itemize}
The remaining sections of this paper are organized as follows: Section \ref{sec:2} illustrates the system model of BSSSQS. The different components of BSSSQS are also discussed in this section. In Section \ref{sec:3}, different transactions of BSSSQS are discussed in details. A security analysis along with performance comparisons among BSSSQS and others existing model is demonstrated in Section \ref{sec:4}. Finally, Section \ref{sec:5} draws a conclusion from this paper with future research directions. 

\section{Proposed Blockchain based System Model}
\label{sec:2}
We devise a QS scheme which uses blockchain concept in order to make it secure and smart. The proposed BSSSQS is a new way of sharing questions. There are four major entities such as Question Setter (QUS), Question Cloud (QC), BSSSQS Master (BSSSQS$_{\textup{master}}$), and BSSSQS Minion (BSSSQS$_{\textup{minion}}$) involved in this scheme\footnote{Note that the main task of Blockchain cloud in Fig. \ref{fig:fig1} is to create communication channels by which BSSSQS$_{\textup{master}}$ and BSSSQS$_{\textup{minion}}$ can communicate with each other.}, as shown in Fig. \ref{fig:fig1}. Each of the entities is described in the following subsections.
\begin{figure*}[ht]
	\centering
	\includegraphics[width=1\textwidth]{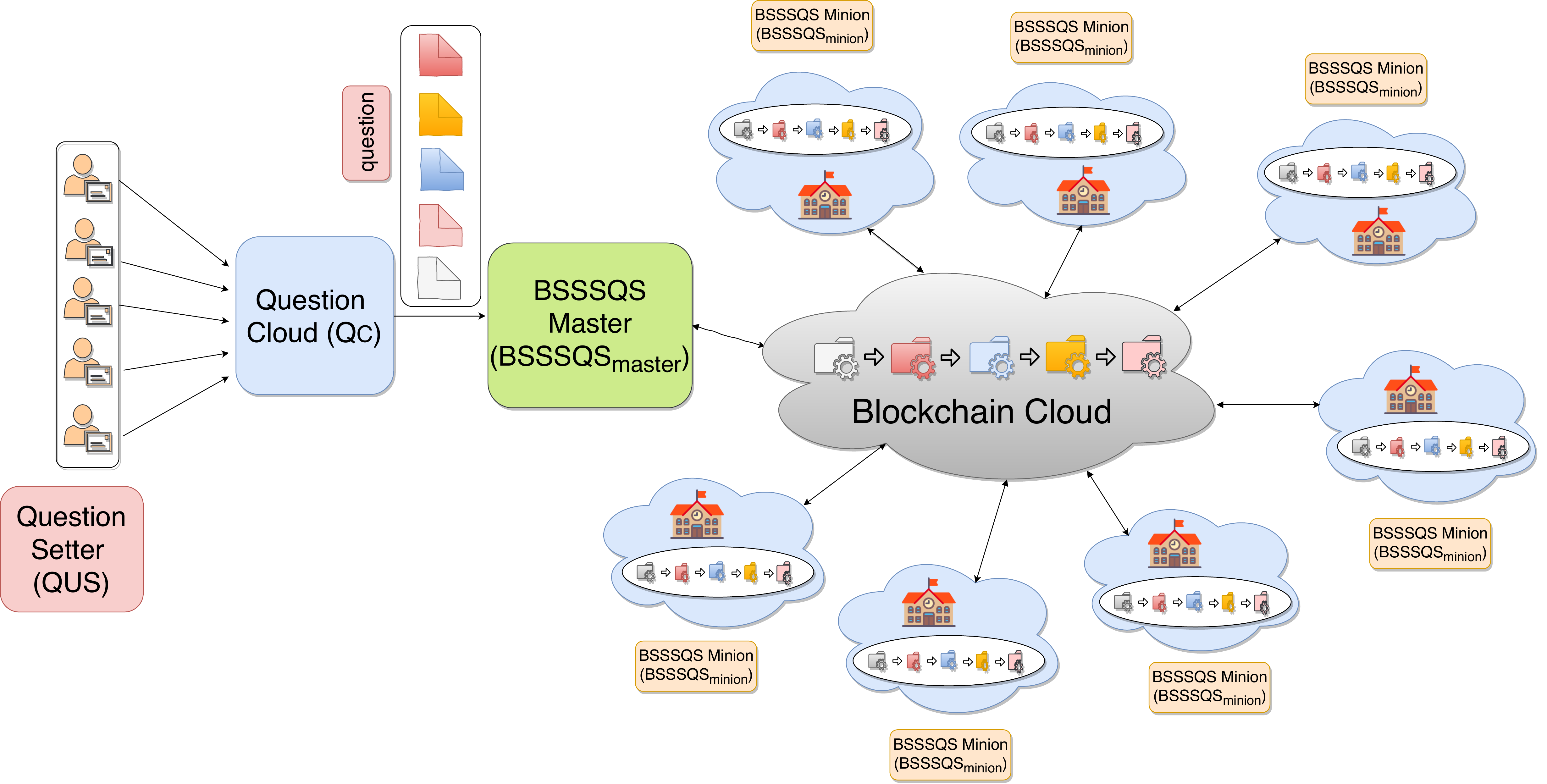}
	\caption{System model of the proposed BSSSQS.}
	\label{fig:fig1}
\end{figure*}

%

\begin{figure}
{\caption{Components of BSSSQS.}\label{fig:fig2}}
{\includegraphics[width=0.3\textwidth]{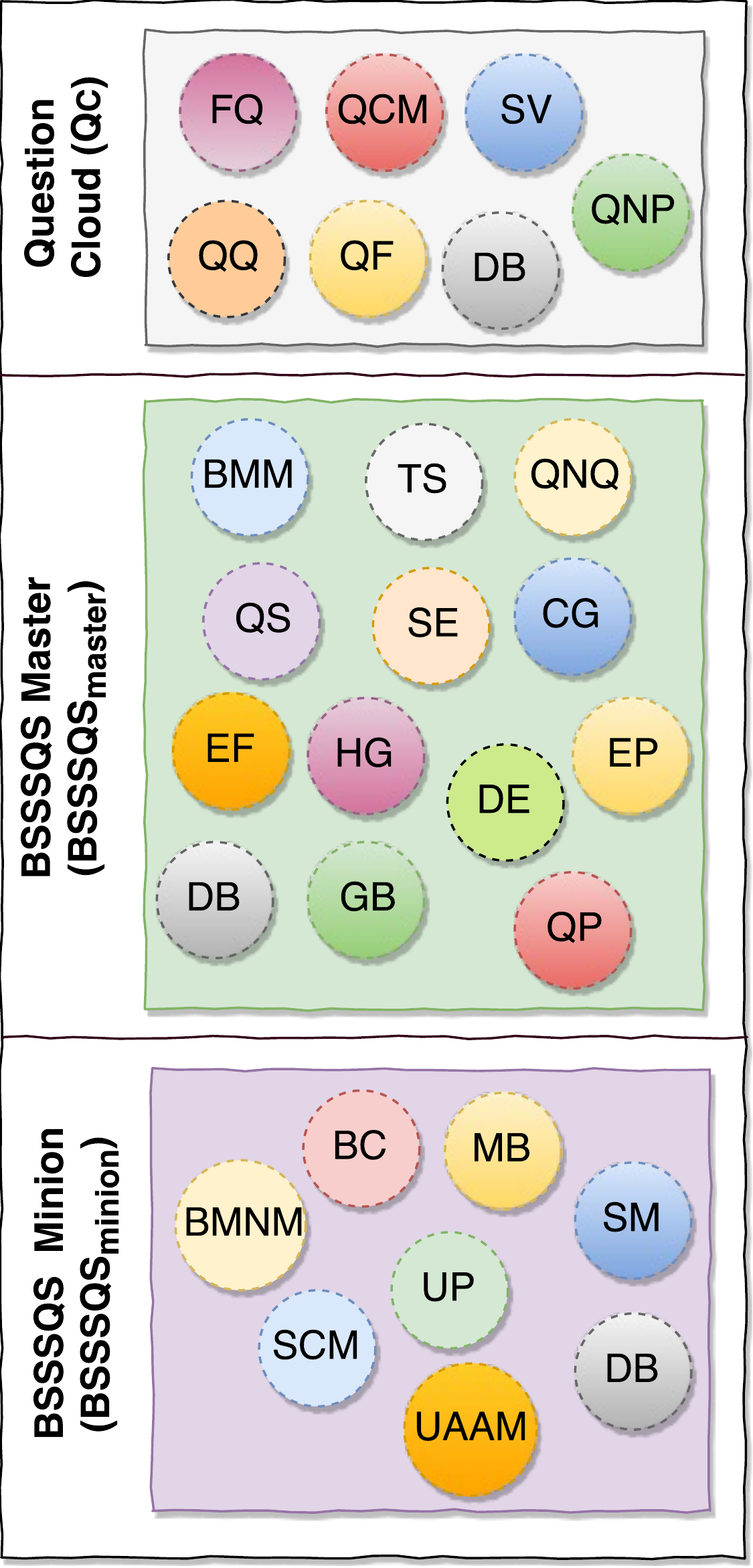}}
\end{figure}

	\subsection{Question Setter (QUS)} In this entity, actors have to submit questions for the exam. They have a deadline for submitting questions. They are the initial actors of preparing questions for the exam. They can modify or delete questions before submitting. But once they submit the question, they lose the option of performing any kind of CRUD (Create, Read, Update, and Delete~\cite{51}) operation. Every actor in this entity has a unique signature to communicate securely with the next entity called question cloud (QC).	
	\subsection{Question Cloud} This entity involves initial management for questions. After getting questions from QUS over a secure channel, it stores questions before sending it to the next entity. In this entity, questions experience modification including shuffling, adding, dropping etc. to prepare QSPs, where a QSP consists of a set of questions. Then, the  QSPs are handed over to the next entity. Note that after handing over to the next entity, QSPs are automatically removed from QC. This entity consists of seven modules, as shown in Fig. \ref{fig:fig2}.  The functions of each module are summarized as below.		
    \begin{itemize}
         \item Question Cloud Manager (QCM): Manages internal functionality.
        \item Signature Verifier (SV): Verifies signatures of requesters.
        \item Format Question (FQ): Formates and modifies question if it requires.
         \item Question pool (QNP): Stores modified questions.  
         \item Question Filter (QF): Sorts and makes sets of questions before sending to next entity.
        \item Question Queue (QQ): Stores questions temporarily before sending to next entity.
        \item Database (DB): Contains signature, course information etc.
            
    \end{itemize}    
	\subsection{BSSSQS Master (BSSSQS$_{\textup{master}}$)} This is a very complex and crucial entity. Questions experience final modification in this entity. This entity holds the information of all the connected minions (node) to which it sends QSPs. This entity also maintains communication with its minions through blockchain cloud. If any of the minions need verification or any kind of approval to exercise any action, BSSSQS$_{\textup{master}}$ responds right away. It not only maintains communication but also selects a QSP, which minion used to take the exam based on that picked out QSP. This entity is also very important for its encryption mechanism. It performs a two-phase encryption in each of the QSPs and creates a smart contract for BSSSQS$_{\textup{minion}}$. This entity consists of thirteen modules, as shown in Fig. \ref{fig:fig2}. The functions of each module are summarized as below.	
	\begin{itemize}		
		\item Question Queue (QNQ): Stores QSPs temporarily.
		\item BSSSQS Master Manager (BMM): Manages internal functionality.
		\item Timestamp (TS): Converts date and time to timestamp.
		\item Question Set (QS): Organizes QSPs based on the course list.
		\item Salt Engine (SE): Generates random salt hash for encryption.
		\item Data Encryptor (DE): Encrypts QSPs based on the selected timestamp. 
		\item Encryption Factory (EF): Encrypts QSPs based on the selected parameters.
		\item Hash Generator (HG): Generates hash of QSPs based on the content, last access time, nonce etc.
		\item Contract Generator (CG): Generates smart contract based on the input.
		\item Database (DB): Stores data of courses, minions, question hashes etc.
		\item Guffy Bot (GB): Monitors tasks and also waits for instructions like re-establish connection with inactive minions, select a QSP etc.
		\item Question Picker (QP): Selects a QSP for the exam.	
		\item Exclusion Pool (EP): Stores QSPs which gets illegal requests.	
	\end{itemize}	
	\subsection{BSSSQS Minion (BSSSQS$_{\textup{minion}}$)} This is an edge entity which contains exam center, e.g. school, college, university etc. This entity contains processed QSPs in the blockchain. No one can access QSPs without experiencing smart contract, timestamp verification, etc. Even though anyone manages to access blockchain storage, he/she can not see the contents of QSPs, because  QSPs are encrypted using a two-phase encryption technique in the proposed BSSSQS.
	This entity consists of eight modules, as shown in Fig. \ref{fig:fig2}. The functions of each module are summarized as below.	
	\begin{itemize}		
		\item BSSSQS Minion Manager (BMNM): Manages internal functionality.
		\item Blockchain (BC): Blockchain based storage which contains QSPs.
		\item Minion Bot (MB): Monitors internal activity like system activeness status (online/offline), crosscheck request validation etc.
		\item Smart Contract Manager (SCM): Handles authorization requests and decrypts QSPs.
		\item Database (DB): Contains decrypted QSPs and other local data.
		\item User Panel (UP): Provides user interface and manages tasks.
		\item Session Manager (SM): Contains information related to user activeness and authorization.
		\item User Authentication and Authorization Manager (UAAM): verifies user authentication and  and manages user authorization.	
	\end{itemize}
\subsection{BSSSQS$_{\textup{master}}$ vs BSSSQS$_{\textup{minion}}$}
BSSSQS contains private blockchain which has different kinds of access permission for different entities. BSSSQS$_{\textup{master}}$ has write permission, whereas BSSSQS$_{\textup{minion}}$ has only conditional read permission. BSSSQS$_{\textup{minion}}$ can access blockchain only through smart contract when BSSSQS$_{\textup{master}}$ sends permission notification to BSSSQS$_{\textup{minion}}$. Note that the final version of QSPs is stored only in the minion entities, though master entity keeps track (e.g., hash) of QSPs.
{\renewcommand{\arraystretch}{1.2}
	\begin{table}[h!]
	\centering
		{\caption{ Notations and their description}\label{table:1}}
		{\begin{tabular}{c|l} 
			\hline
			Notation & Description \\ [0.5ex] 
			\hline\hline
			$\eta$ & Nonce  \\ 
			$\rho$ & Prime number \\
			$Q$ & Question  \\
			$\tau_c$ & Current timestamp \\
			$\mathbb{S}_\hbar$ & Salt hash\\ 
			$QT$ & Questionnaire token  \\			
			$PW$ & Password \\
			$\Theta(.)$ & One way key generation function \\
			$\xi_{key}(.)$ & Encryption function using key \\			
			$\zeta_{key}(.)$ & Decryption function using key \\
			$\mathbb{Q}$ & Encrypted question\\ 
			$sm$ & Smart contract  \\ [1ex]
			\hline
		\end{tabular}}	
\end{table}

\section{Transactions in BSSSQS}
\label{sec:3}

In this section, we describe the different types of transactions performed in BSSSQS. The list of important notations with related descriptions which are used in this section, are summarized in Table \ref{table:1}.

	\subsection{Transactions between QUS and QC}

\begin{figure}
{\caption{Transactions between QUS and QC.}
		\label{fig:fig3}}
{\includegraphics[width=0.6\textwidth]{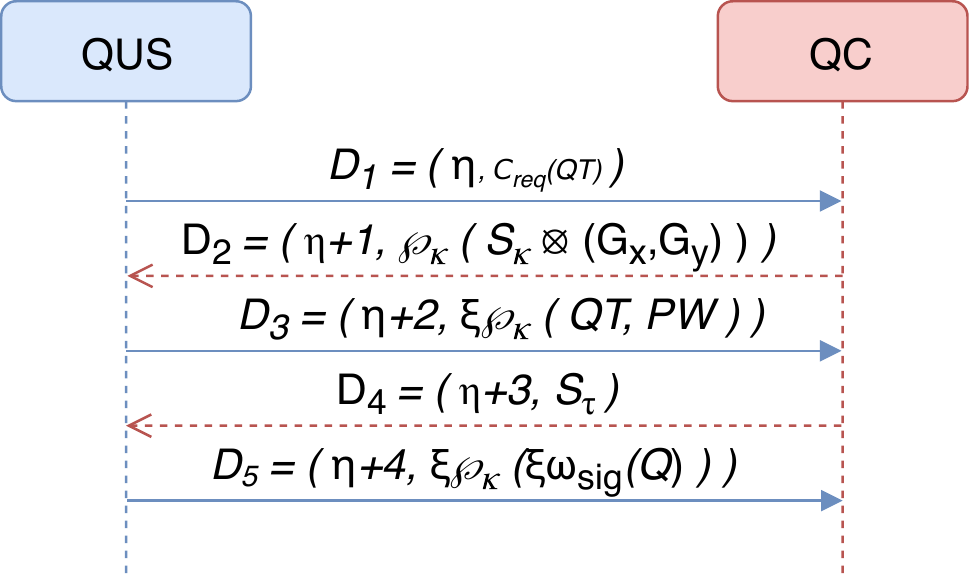}}
\end{figure}


	Basically two types of major transactions, such as authentication of QUS and questions hand over to QC, take place between QUS and QC. Every user in QUS has a unique signature, which is stored in QC. Each user has to prove his identity with proper credentials to send questions to QC.
The transactions are provided in the Fig. \ref{fig:fig3}. The proposed scheme considers that all communication between QUS and QC is done by employing asymmetric key encryption (i.e. elliptic curve digital signature algorithm (ECDSA)). Before sending questions, QUS sends a request to QC in order to get public key of QC so that QUS can communicate with QC in a secured channel. QUS sends request, containing a nonce $\eta$, by sending data $D_1 = ( \eta, \mathbb{C}_{req}(QT))$. After getting request from QUS, QC generates one time asymmetric keys (OTAK) for QUS so that QUS can transfer not only credentials but also questions employing OTAK. QC generates a secret key $S_{\kappa}$ and a public key $\wp_{\kappa}$. Let $\overline{\kappa}$ is the set of OTAK.

\begin{equation}
\label{eq:key}
\begin{gathered}
 \overline{\kappa} = \{ (\varphi,\psi) \; : \; \varphi \in S_{\kappa} \cap \psi \in \wp_{\kappa} \;  | \\ \; S_{\kappa} = \Theta((\eta + 1)* \rho, \tau_c,  QT, \mathbb{S}_\hbar) \cap \wp_{\kappa} = S_{\kappa} \otimes (G_x, G_y) \} 
\end{gathered}
\end{equation}

Here, $\rho$ is a large prime number, $\tau_c$ is current timestamp, $\mathbb{S}_\hbar$ is a salt hash and $G$ is a set of $(x, y)$ coordinates on the elliptic curve. When QC finishes generating OTAK, QC sends $\wp_{\kappa}$ by sending $D_2 = (n+1, \wp_{\kappa} = S_{\kappa} \otimes (G_x, G_y))$. Upon receiving $D_2$, QUS encrypts $QT$ and $PW$ employing $\xi_{\wp_{\kappa}}(QT,PW)$. Following this, QUS increases $n+1$ by 1 and sends data $D_3$ to QCM. When QC receives $D_3$, QC first decrypts data utilizing $\zeta_{S_{\kappa}}(\xi_{\wp_{\kappa}}(QT,PW))$ and checks validity of the credential that provided from QUS. If the credential is valid then QC returns  success token by sending $D_4 = ( \eta+3, S_\tau)$. Before sending questions, QUS needs to sign question with its digital signature in order to prove that the questions are coming from the intended person. First, QUS creates signature using $QT$ and $PW$. Let $\omega_{sig}$ is the signature.

\begin{gather}
\label{ques:sig}
 \omega_{sig} = \Theta(QT,PW) 
\end{gather}

QUS first encrypts question using $\omega_{sig}$ and then encrypts using $\wp_{\kappa}$. After the encryption, QUS sends $D_5 = ( \eta+4, \xi_{\wp_{\kappa}}(\xi_{\omega_{sig}}(Q)) )$. As QC receives data from QUS, QC decrypts $\xi_{\wp_{\kappa}}(\xi_{\omega_{sig}}(Q))$ by employing $\zeta_{S_{\kappa}}(\xi_{\wp_{\kappa}}(\xi_{\omega_{sig}}(Q)))$. QC generates signature employing Eq. (\ref{ques:sig}) and validates the identity of the person by decrypting question with sender's signature. 

	\subsection{Transactions between QC and BSSSQS$_{\textup{master}}$}	
	
	Here, transactions are divided into two main categories, as shown in Fig. \ref{fig:fig4}, such as 1) processing questions within different modules of QC, and 2) sending QSPs from QQ to BSSSQS$_{\textup{master}}$ for further processing. After the deadline of questions submission, FQ formates and modifies questions following the prescribed rules and regulations set by the exam authority, to prepare QSPs. Questions formatting and modification can contain actions like shuffling contents, adding or removing contents etc. Only the authorized person from exam committee can do that or the exam authority can make this process automatic using randomization and weighted policy. The QSPs are then sent to QNP. In QNP, questions are stored temporarily and wait for next instruction. After getting proper instructions from QCM, QSPs are sent to QF. QF selects some QSPs based on certain criteria and forwards these selected QSPs to QQ for gathering them before sending to BSSSQS$_{\textup{master}}$. When the collection is finished, QQ sends QSPs to BSSSQS$_{\textup{master}}$ through a proper secure channel.
	\begin{figure}
		\centering
		\includegraphics[width=1\textwidth]{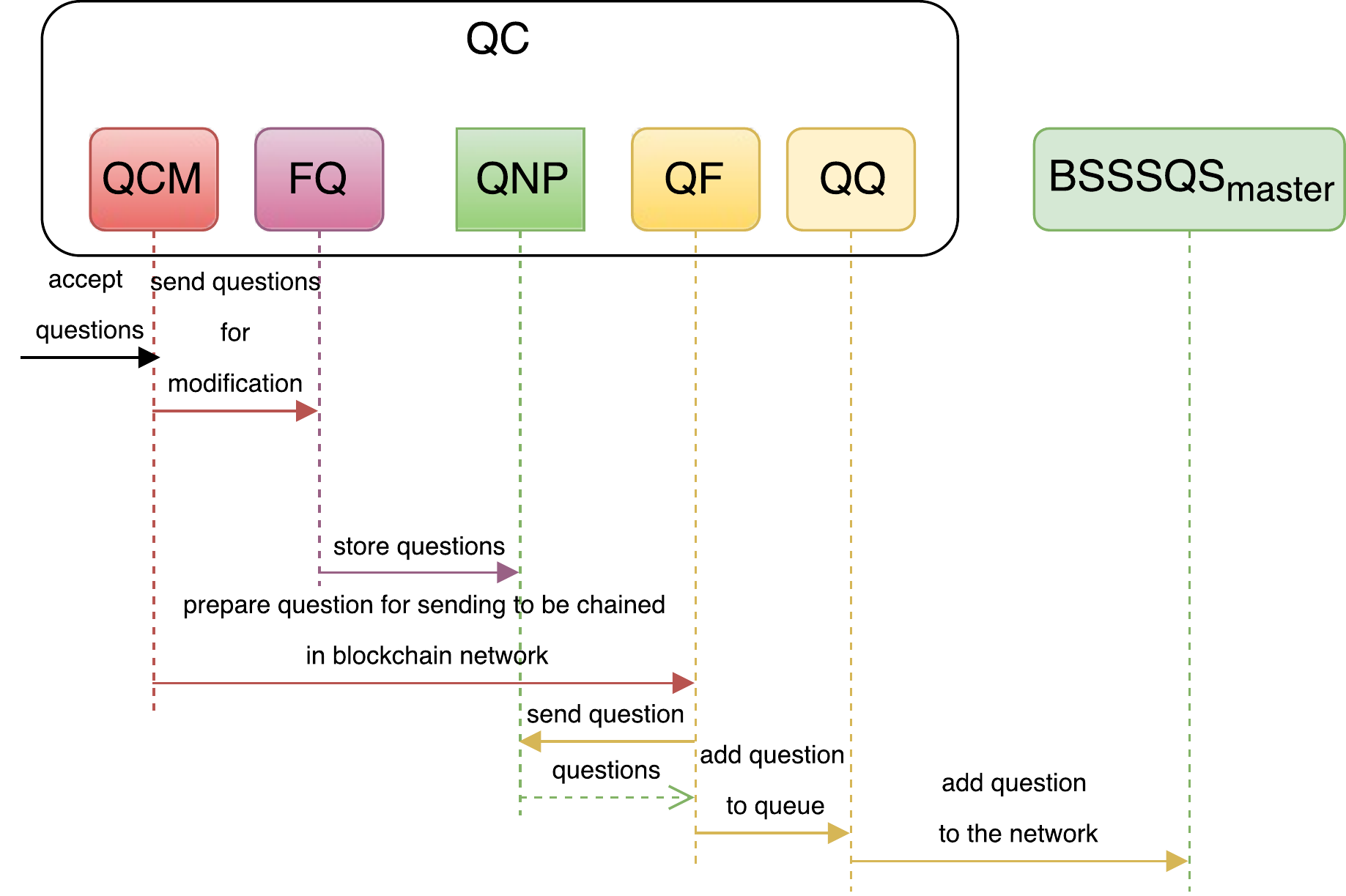}
		\caption{Transactions between QC and BSSSQS$_{\textup{master}}$.}
		\label{fig:fig4}
	\end{figure}	
	\subsection{Transactions between BSSSQS$_{\textup{master}}$ and BSSSQS$_{\textup{minion}}$}

	This segment covers transactions between BSSSQS$_{\textup{master}}$ and  BSSSQS$_{\textup{minion}}$, as shown in Fig. \ref{fig:fig5}. The primary tasks of BSSSQS$_{\textup{master}}$ are summarized as follows:
	\begin{itemize}
		\item To encrypt QSPs and send these encrypted QSPs to BSSSQS$_{\textup{minion}}$.
		\item To select a QSP for the exam and send permission notification to BSSSQS$_{\textup{minion}}$ for accessing the selected QSP.
	\end{itemize}
	BSSSQS$_{\textup{master}}$ plays very significant roles for providing security to QSPs. The proposed scheme considers that all communication between BSSSQS$_{\textup{master}}$ and BSSSQS$_{\textup{minion}}$ is done by employing symmetric key encryption. Initially, questions are stored in QNQ. After getting QSPs from QNQ, BMM picks current timestamp $\tau_c$ by sending request to TS. In the subsequent stage, BMM pulls course list from DB so that it can command other modules to make sets of blocks containing QSPs based on the course. Next, BMM sends QSPs to QS with $\tau_c$ and course list. Then, QSPs experience two-phase encryption which is described as follows.

	\subsubsection{First-phase encryption} The first phase of encryption is managed by QS.
	 Firstly, QS requests SE for generating salt hash $\mathbb{S}_\hbar$. After getting $\mathbb{S}_\hbar$ from SE, QS stores it for next phase of encryption. Secondly, QS sends QSPs to DE with $\tau_c$. DE is then encrypts QSPs with $\tau_c$. Let $Q_i$ is the $i^{th}$ number of QSPs. Therefore, $i^{th}$ number of QSPs that experience the first phase of encryption is written by 
	 \begin{equation}
	 \label{eq:eq1}
		 \mathbb{Q}^1_i = \xi_{\tau_c} \{Q_i, \tau_c \} 
		\end{equation}
	Finally, QS sends encrypted QSPs to EF with $\tau_c$ and $\mathbb{S}_\hbar$.
	\subsubsection{Second-phase encryption} 
	 The second and final phase of encryption happens in EF. Recall that in the blockchain, the first block is called \textit{genesis block}. EF generates a default genesis block with random text and encrypt it with $\tau_c$. After creating genesis block, EF encrypts QSPs and converts these QSPs  into blocks. Each block contains header and data. In the header, it carries the previous block hash, timestamp, last access time, block creation time, and nonce. Every time EF encrypts a QSP, it sends that encrypted QSP to HG. HG is then generate a hash from that encrypted QSP, so that the hash can participate in the next QSP encryption. Therefore, the encrypted QSP that experience the second phase of encryption is written by 	 
	 \begin{figure*}
		\centering
		\includegraphics[width=0.95\textwidth]{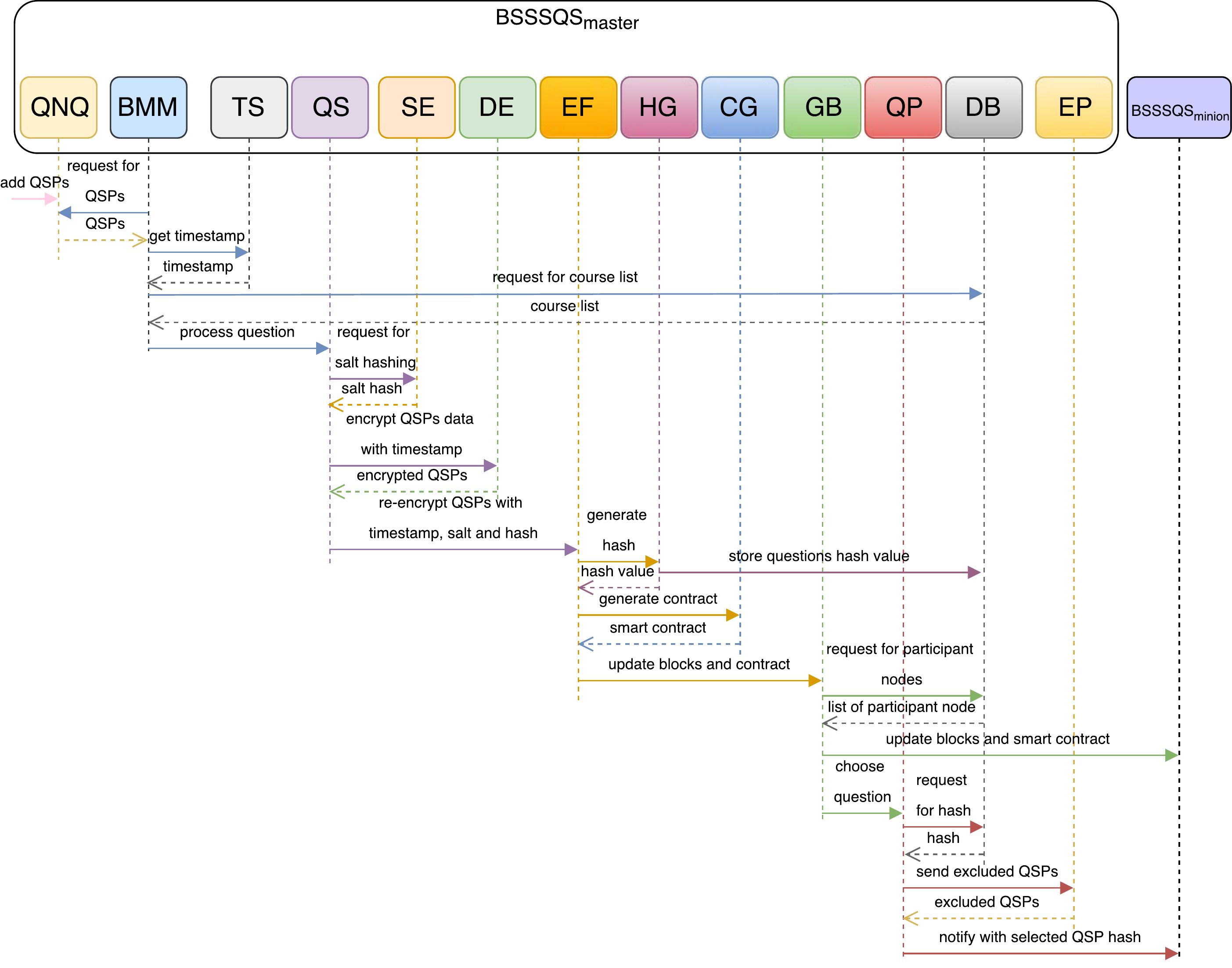}
		\caption{Transactions between BSSSQS$_{\textup{master}}$ and BSSSQS$_{\textup{minion}}$.}
		\label{fig:fig5}
	\end{figure*}
	 \begin{equation}
	 \label{eq:1}
	 \mathbb{Q}^2_i =\left\{
	 \begin{array}{@{}ll@{}}
	 \xi_{\tau_c} \{\mathbb{Q}^1_i, \tau_c \}, & \text{if}\ i=0 \\ \\
	 \xi_{\{\tau_c,\mu_{\mathbb{Q}^\hbar}, \mathbb{S}_\hbar \}} \{\mathbb{Q}^1_i, \tau_c,  \mu_{\mathbb{Q}^\hbar}, \mathbb{S}_\hbar \}, &  \text{if}\ i > 0
	 \end{array}\right.
	 \end{equation}
	 where $ \mu_{\mathbb{Q}^\hbar} = \bigcup_{\sigma=0}^{i-1} \mathbb{Q}^{\hbar}_\sigma$. The hash of $i^{th}$ QSP is generated by 
	  \[ \mathbb{Q}^{\hbar}_i = \mathbb {SHA}256(\mathbb{Q}^2_i) \]
    Then, HG stores the generated hash in DB for QSP selection. After that, EF commands CG to generate smart contract including the information of $\mathbb{Q}^\hbar$, $\tau_c$, and $\mathbb{S}_\hbar$. Fig. \ref{fig:fig6} illustrates a coding structure for smart contract. Smart contract contains hashes of QSPs, timestamp, and salt hash.

%
    As smart contract is created dynamically based on the exam, so $id$ and $title$ of contract title $exam\_\{id\}\_\{title\}$ are different for each exam as shown in Fig. \ref{fig:fig6}. When smart contract generation is finished, CG encrypts the smart contract with a timestamp $\tau_{sm}$ and a random salt hash $\mathbb{S}_r$. Let $\mathbb{C}_{sm}$ is the encrypted smart contract.
  
\begin{gather}	
\label{smart:enr}   
    \mathbb{C}_{sm} = \xi_{\{ \Theta(\tau_{sm},\mathbb{S}_r) \}} \{sm,\Theta( \tau_{sm}, \mathbb{S}_r) \} 
\end{gather}
   
   After the encryption, CG stores the key in DB. When exam comes, \\\text{BSSSQS$_{\textup{master}}$}  sends the key along with a selected question hash. After getting the encrypted smart contract from CG, EF sends blocks and smart contract to GB. GB monitors tasks like re-establish connection with inactive minions, select a QSP etc. As GB gets blocks and contract, it initiates the process of sending these resources to BSSSQS$_{\textup{minion}}$. At first, GB pulls existing minion list from DB. When GB get all of the lists, it starts to send blocks and contract to  BSSSQS$_{\textup{minion}}$ through blockchain cloud. GB remains always active for monitoring the system. When the time comes to select a QSP for exam, GB sends an instruction to QP for initiating the process of selecting a QSP for the exam along with $\mathbb{C}_{sm}$ and notifying minions about that QSP. Before initiating random engine for picking out a QSP, QP pulls hash of QSPs from DB. In the meantime, it also requests EP to send the hashes of excluded QSPs. When QP gathers all the required information, it starts the process of selecting a QSP as follows. Firstly, QP removes the excluded QSPs from the set of QSPs. Therefore, the set of filtered QSPs $\overline{Q_F}$ is written by 
\begin{gather}	
\label{qsp:excl}
 \overline{Q_F} = \{ \chi \; : \; \chi \in Q_P \cap \chi \notin Q_E  \}
\end{gather}
    	
    	where $Q_P$ is the set of all QSPs and $Q_E$ is the set of excluded QSPs.
    	Secondly, QP takes a collection of $10$ large prime numbers which is represented by 
\begin{equation}
\label{eq:prime}
\begin{gathered}
    	 \overline{\rho} = \{ \rho_i \in \mathbb P \; | \; 0 \leq i \leq 9 \} \; \\  where \; \mathbb P = \{ \chi \; : \; \chi \in \mathbb{N} \; \cap \; \chi \; is \; prime \}     	 
\end{gathered}
\end{equation}

\begin{figure}
\includegraphics[width=0.8\textwidth]{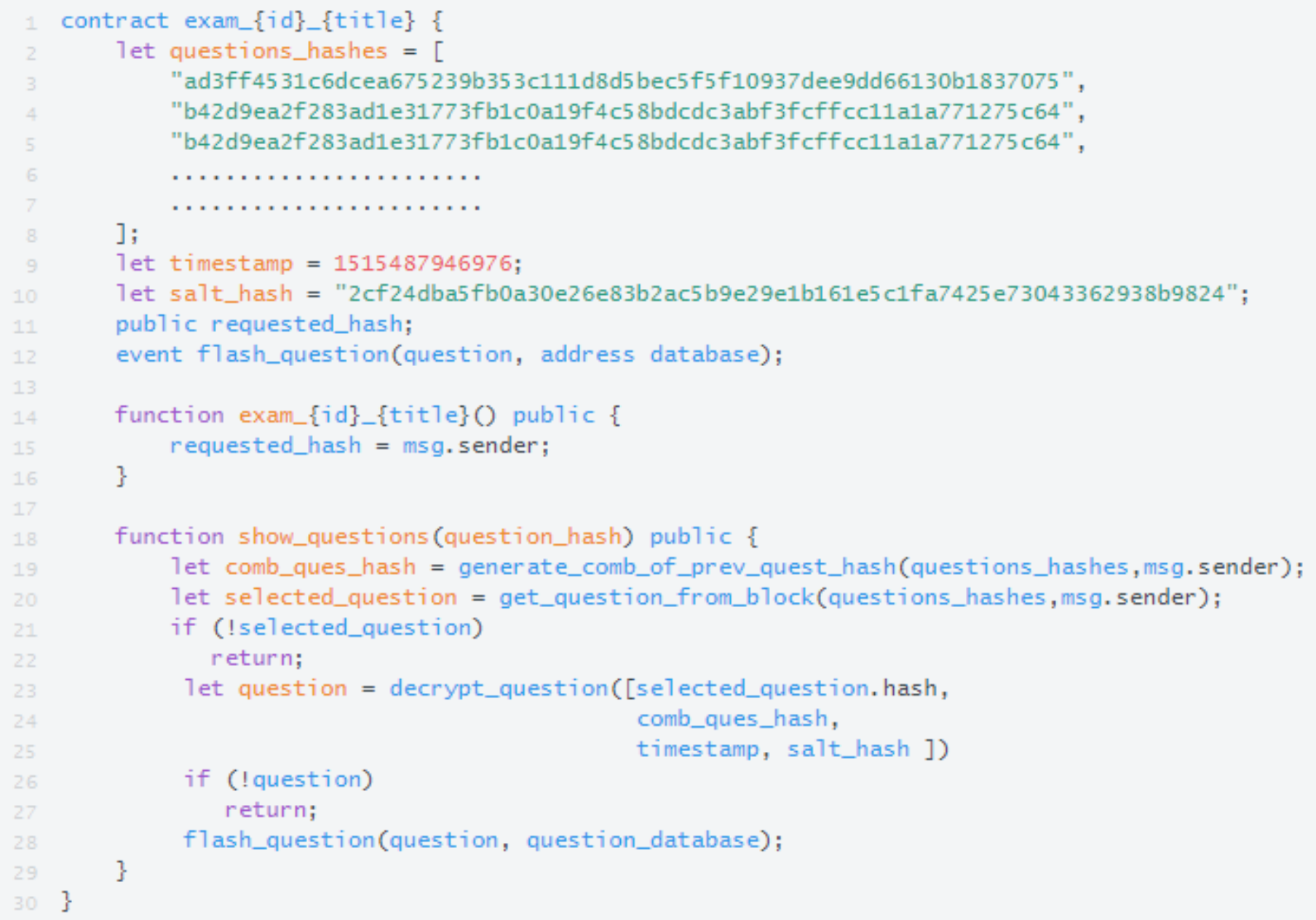}
\caption{Proposed coding structure of smart contract. }
\label{fig:fig6}
\end{figure} 

    	After that it converts the current date and time into a timestamp $\tau$. To select any two prime numbers from $\overline{\rho}$, QP takes last digit $d_\tau^l$ and second last digit $d_\tau^{sl}$ from $t$ to select prime numbers $p_{d_\tau^l}$ and $p_{d_\tau^{sl}}$, respectively. The selected QSP to take the exam is represented by 
\begin{gather}
\label{eq:qspsel}
 Q_s = \{(p_{d_\tau^l} - Q_{fn})*p_{d_\tau^{sl}}\}\mod Q_{fn}
\end{gather}
where $Q_{fn}$ is the total number of filtered QSPs and $p_{d_\tau^l} \geq Q_{fn}$. Let $\overline{\rho}  = \{ 179426549, 24066347, 179424793, 15486511, ...., 17142407 \}$, $\tau = 1515552555821$, $d_\tau^l = 1$, $d_\tau^{sl} = 2$, $p_{d_\tau^l} = 24066347$, $p_{d_\tau^{sl}} = 179424793$, and $Q_{fn} = 50$. The value of $Q_s$ is then computed as $Q_s = \{(24066347 - 50) * 179424793\}\mod 50 = 21$. It means that $21^{th}$ QSP is selected from the collection of  $50$ QSPs for the exam. As QP selects a QSP, it notifies all BSSSQS$_{\textup{minion}}$ about the selection through blockchain cloud.    	
	\subsection{Transactions in BSSSQS$_{\textup{minion}}$}	
	\begin{figure*}
		\centering
		\includegraphics[width=0.90\textwidth]{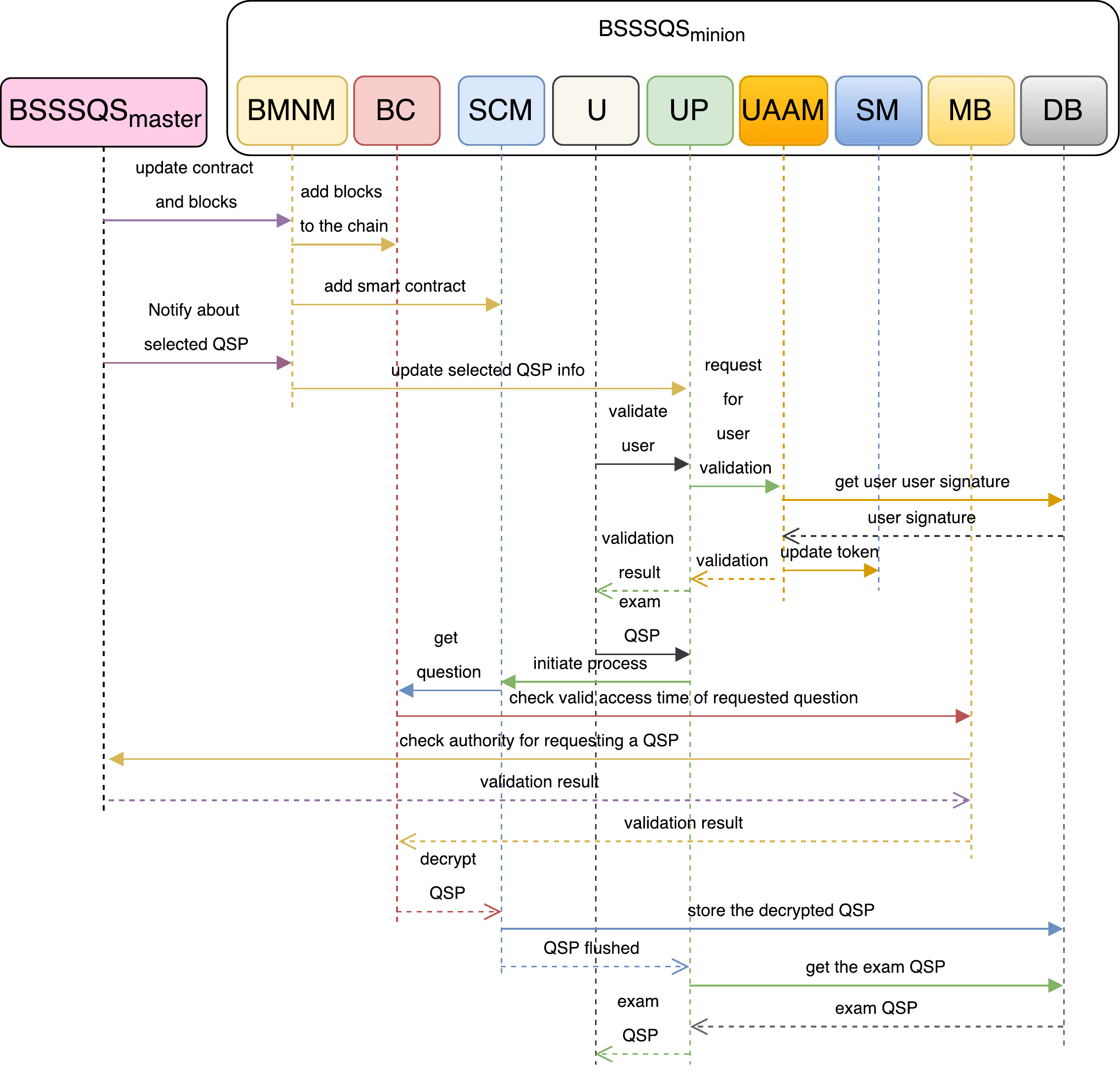}
		\caption{Transactions in BSSSQS$_{\textup{minion}}$.}
		\label{fig:fig7}
	\end{figure*}	
	This section covers transactions between different modules of BSSSQS$_{\textup{minion}}$ which is illustrated by Fig. \ref{fig:fig7}. Note that U in Fig. \ref{fig:fig7} represents a user in the system. The transactions are broadly categorized into three following  types.  
	\begin{enumerate}
		\item Storing and maintaining QSP blocks in blockchain. 
		\item Updating smart contract and 
		\item Alerting authority about the permission to access QSPs.
	\end{enumerate}
	 After getting blocks and smart contract from BSSSQS$_{\textup{master}}$, the following steps are performed sequentially in BSSSQS$_{\textup{minion}}$ as follows:
	 \begin{enumerate}
	 	\item BMNM sends blocks to BC and BC updates his data. 
	 	\item BMNM sends smart contract to SCM for the selected exam.
	 	\item BMNM gets QSP selection notification from BSSSQS$_{\textup{master}}$. The meaning of notification is that user can request for the process of decrypting the selected QSP. But If a user tries to request before allowed time, SCM doesn't transmit the request to blockchain. However, BMNM passes this notification to UP and UP alerts users when they enter the system.
	 	\item When a user tries to enter UP, he has to go through a validation process. If the session of that user is expired then UP requests UAAM to corroborate the user. UAAM sends a request to DB to send information regarding the requested signature. If the user is valid, DB returns user information, otherwise, it reruns empty data. When UAAM gets validation from DB, it stores a token in SM for maintaining user session.
	 	\item Every minion manages its users by itself. After that, UAAM notifies UP about the response. As users get a notification about the QSP and key for decrypting $\mathbb{C}_{sm}$, they request for QSP through UP. UP requests SCM to start the decryption process. Before going further, SCM sends a command to MB to check whether QSP is unlocked for access or not.
	 	\item  MB affirms authorization with BSSSQS$_{\textup{master}}$ through blockchain cloud. When SCM gets proper authorization, it transfers the request to BC. BC performs final authorization check with BSSSQS$_{\textup{master}}$ through MB. If BC gets an unauthorized request with a QSP, it changes access time and nonce, and mines the chain again. It changes the signature of all the QSPs and no one can get its key hash. 
	 	\item Whenever BC gets an affirmative result, it sends the QSP to SCM for decrypting.
	 \end{enumerate} 
   
   First, SM decrypts $\mathbb{C}_{sm}$ in order to decrypts the selected QSP. Let $D_{sm}$ is the decrypted smart contract.
   \begin{gather}
\label{eq:dec1}
    D_{sm} = \zeta_{\{ \Theta(\tau_{sm},\mathbb{S}_r) \}} \{\mathbb{C}_{sm},\Theta( \tau_{sm}, \mathbb{S}_r) \} 
   \end{gather}
   
   After getting $D_{sm}$, QSP decryption process begins. Let $\mathbb{Q}$ is a selected encrypted QSP. By utilizing  $ \mu_{\mathbb{Q}^\hbar}$ in Eq. (\ref{eq:1}), the first phase of decryption is written by  
	\begin{gather}
\label{eq:dec1}
	 D_{Q_1} = \zeta_{\{\tau_c,\mu_{\mathbb{Q}^\hbar}, \mathbb{S}_\hbar\}} \{\mathbb{Q},\tau_c, \mu_{\mathbb{Q}^\hbar}, \mathbb{S}_\hbar\} 
	\end{gather}
	
	Finally, $D_{Q_1}$ goes through the second phase of decryption which is written by 
	\begin{gather}
\label{eq:dec2}
	 D_{Q_2} = \zeta_{\tau_c} \{D_{Q_1},\tau_c\} 
	\end{gather}
	
	where $D_{Q_2}$ is the QSP which experiences the second phase of decryption. Note that the decrypted QSP is the original form of the QSP. After decryption, SC stores the QSP to DB and sends a notification to UP about the outcome. Now, users can access the QSP in original form. Finally, users can retrieve QSP from DB to take an exam.

\section{ Security and performance analysis}
\label{sec:4}

In this section, we discuss the security and performance of the proposed BSSSQS in order to illustrate to the feasibility of  BSSSQS.

\subsection{Security analysis}

In this section, we propose different propositions related to security and performance and also provide proof of each proposition.

\begin{proposition}
\label{prop:1}
The secret key of QC is well protected from the adversary.
\end{proposition}

\begin{proof}
Qc's secret key is generated utilizing nonce added by $1$ which is $\eta+1$, a large prime $\rho$, timestamp $\tau_c$ of that time, $QT$ and a random $\mathbb{S}_\hbar$, as mentioned in Eq. (\ref{eq:key}). Suppose, adversary A wants to steal the secret key of QC. Only one way to get the secret key of QC is to guess the private key to the best of our knowledge, as QC never share its private key to anyone. However, in ECDSA, the secret key is $32$ bytes or $256$ bits long. In order to guess the correct secret key, A needs to guess the sequence of $256$. For $256$ bits, there are $2^{256}$ possible sequences and among of them, only one can be the QC's secret key.  The probability of guessing the secret key, which is $256$-bit long, is $\dfrac{1}{2^{256}} = 2^{-256}$, which is practically not feasible. Moreover, if A wants to guess properties of the secret key individually, A has to face the probability of randomness in each property which is also practically not feasible. Furthermore, OTAK is temporary, when questions are transferred successfully, OTAK, which is generated for particular QUS, is removed from QC. Therefore, the secret key may become obsolete while A is still trying to guess the secret key. Thus, QC's secret key is well protected from the adversary.
\end{proof}

\begin{proposition}
Communication between QUS and QC is secure even in the presence of an eavesdropper.
\end{proposition}

\begin{proof}
The motive behind the communication between QUS and QC is to transfer questions from QUS to QC. There is a set of task has to be performed before sending questions. In order to send questions, QUS requires QC's public key to create a digital signature using Eq. (\ref{ques:sig}) and encrypt questions employing $\xi_{\wp_{\kappa}}(Q)$. However, when QUS request for OTAK, QC generates OTAK utilizing Eq. (\ref{eq:key}). When QUS gets $\wp_{\kappa}$, Firstly, QUS validates its identity by transmitting $QT$ and $PW$, which is encrypted using $\wp_{\kappa}$, to QC. Secondly, it generates a digital signature by applying Eq.(\ref{ques:sig}). Finally, QUS encrypts questions using $\wp_{\kappa}$ and sends back to QC signed with its signature. Suppose, there exists an eavesdropper named B between QC and QUS. B wants to steal credentials of QUS along with questions that QUS sends to QC and also wants to send false data to QC. B catches data between QUS and QC and B wants to extract $D_2$ and $D_5$ data, as shown in Fig. \ref{fig:fig3}. In order to extract data, B requires QC's private key and it's not available to anyone except QC. Moreover, there is no feasible solution to extract private key from public key by reverse engineering or guessing, as discussed in \textbf{Proposition \ref{prop:1}}. However, B wants to send false data encrypted by $\wp_{\kappa}$ to QC. But when QUS send questions to QC, QUS signs the question with its signature. From the signature, QC verifies the actual source of the data. As B needs the signature of QUS, B cannot send false data until it obtains QUS's signature and QUS's signature is not only publicly unavailable, but also QUS never share its signature with other people apart from sharing with QC in an encrypted form. As a result, B cannot achieve any of the aforementioned objectives and thus B's activity has no effect on the communication between QUS and QC. 
\end{proof}

\begin{proposition}

QSP selection in BSSSQS$_{\textup{master}}$ is totally random and is free from compromised QSPs.
\end{proposition}

\begin{proof}
Before the exam, BSSSQS$_{\textup{master}}$ selects a QSP and sends that QSP reference to BSSSQS$_{\textup{minion}}$. This process is totally random. Before selecting a QSP,  BSSSQS$_{\textup{master}}$ selects a set of $10$ prime numbers. Each prime number is selected following the uniform distribution, as shown in Eq. (\ref{eq:prime}). Let, $\mathbb(P)$ is the set of prime numbers and $\overline{\rho}$ is the set of already selected prime numbers. So, the probability of selecting prime numbers $P(\mathbb P - \overline{\rho})$. However, after selecting the set of prime numbers, two prime numbers are picked from that set, based on a timestamp value, for further processing. Finally, QSP is selected by employing Eq. (\ref{eq:qspsel}) which gives a random QSP number among the set of QSPs. BSSSQS$_{\textup{master}}$ is well protected scheme. By any chance, if any QSP becomes compromised, BSSSQS$_{\textup{minion}}$ notifies BSSSQS$_{\textup{master}}$ about that question. BSSSQS$_{\textup{master}}$  excludes that compromised QSP from the selection process by employing Eq. (\ref{qsp:excl}) which makes the selection process free from compromised QSPs.
\end{proof}

\begin{proposition}
\label{prop:2}
QSPs and smart contract are secure from physical attacks by both insiders and outsiders.
\end{proposition}

\begin{proof}
Physical attacks are one kind of attack which involves exploiting the weakest point by the attacker to breach the security system. There are different kinds of physical attacks, such as (1) walk-in, (2) break-in, (3) sneak-in, and (4) damage equipment. However, attackers may not always come from outside. Sometimes a person from the inside may also harm the system. As we discussed in Introduction that sometimes teacher or authority may leak the question, so it's very important to give protection from the attack which is caused by both outsider and insider. BSSSQS imposes a timelock on the QSPs and the smart contract. If anyone tries to access both of them before the allowed time, the system notifies not only BSSSQS$_{\textup{master}}$ but also  BSSSQS$_{\textup{minion}}$. Suppose, an attacker X from inside wants to steal QSPs. X disables the connection of BSSSQS$_{\textup{minion}}$ and tries to copy QSPs from disk. QSPs access permission is locked which can be unlocked by the permission from BSSSQS$_{\textup{master}}$. However, X somehow bypasses the access protection and start to copy. An internal monitoring system monitors this activity and changes QSPs auditing time, which changes the hash of QSPs and breaks the chain of the block.  When BSSSQS$_{\textup{minion}}$ comes to online, BSSSQS$_{\textup{minion}}$ notifies BSSSQS$_{\textup{master}}$ and BSSSQS$_{\textup{master}}$ excludes the BSSSQS$_{\textup{minion}}$ from taking exam due to QSPs blocks are not chained together. However, after copying the QSPs, X needs a private key to unlock both QSPs and smart contract which are encrypted by employing  Eq. (\ref{eq:1}) and Eq. (\ref{smart:enr}) respectively. In \textbf{Proposition \ref{prop:1}}, we have discussed that it's not feasible to guess a 256-bit key. Therefore, copying the QSPs is not going to help X. This outcome is the same for the outsider also. In this way, QSP and smart contract are secure from physical attacks by both insiders and outsiders.
\end{proof}

\begin{proposition}

Smart contract is well protected because it contains keys for unlocking QSP.

\end{proposition}

\begin{proof}
While sending QSPs to BSSSQS$_{\textup{minion}}$, a smart contract is sent along with QSPs. QSPs are not only encrypted using their hashes, but also using a timestamp $\tau_c$ and a random hash $\mathbb{S}_\hbar$. In order to decrypt the selected QSP by employing Eq. (\ref{eq:dec1}) and Eq. (\ref{eq:dec2}), BSSSQS$_{\textup{minion}}$ also needs $\tau_c$ and $\mathbb{S}_\hbar$ along with the selected QSP hash. BSSSQS$_{\textup{master}}$ creates a smart contract containing hashes of the QSPs,  $\tau_c$ and $\mathbb{S}_\hbar$. When a QSP's hash code is provided in the smart contract, the smart contract decrypts the provided QSP. Therefore, in order to decrypt a QSP, BSSSQS$_{\textup{minion}}$ has to go through the smart contract because smart contract holds the other two keys. However, before sending QSP, BSSSQS$_{\textup{master}}$ encrypts the smart contract using Eq. (\ref{smart:enr}). According to \textbf{Proposition \ref{prop:2}}, the smart contract is protected using a timestamp lock along with an encryption technique, and breaking the security is not feasible. Thus, we can say that the smart contract is well protected and holds keys to unlock the QSP.
\end{proof}

%

\begin{proposition}
QSPs in BSSSQS$_{\textup{minion}}$ are tamper-proof.
\end{proposition}

\begin{proof}
QSPs in BSSSQS$_{\textup{minion}}$ are stored in blockchain. In blockchain, data are stored in blocks which are linked together~\cite{1}. For example, $(n+1)^{th}$ block holds the hash of $n^{th}$ block. In blockchain, a block contains a nonce, timestamp, transactions and a Merkle tree which contains hashes of the transactions. A hash is the identity of a block which is generated from aforementioned properties.
If any value in the transaction changes, the value of Merkle tree also changes. As Merkle tree value changes, the hash of the block also changes as the hash is totally dependent on the value of each property. In this way, if any illegal changes occur, it can be detected very easily in the blockchain. However, any illegal activities on QSPs are also detectable as QSPs are stored in blockchain. Thus, QSPs become tamper-proof. 
\end{proof}

\begin{proposition}
QSPs are distributed among  all BSSSQS$_{\textup{minion}}$ in order to reduce the load of BSSSQS$_{\textup{master}}$.
\end{proposition}

\begin{proof}
Blockchain is a distributed ledger which is distributed among the participating nodes and each participant holds the same copy of the data~\cite{1}. BSSSQS applies this concept in order to distribute the QSPs to BSSSQS$_{\textup{minion}}$ which reduces its load from distributing QSPs at the exam time and to maintain the security of these QSPs remotely. When it's time for the exam, \\BSSSQS$_{\textup{master}}$ just passes the hash of a QSP and a secret key for unlocking the smart contract instead of a whole QSP which reduces the load on the system. However, QSPs are occupying storage space of BSSSQS$_{\textup{minion}}$ and with the increase of QSPs, storage space also occupies very quickly, which may create issues in the space of the storage. In order to mitigate that, when the exam finishes, BSSSQS$_{\textup{minion}}$ removes QSPs from its storage and makes space free. With the help of blockchain, it becomes very easy to manage QSPs securely in a distributed way. 
\end{proof}

\subsection{Performance comparison}

A comparative study between BSSSQS and existing models, such as Chang et al.\cite{44}, Kaya et. al\cite{45}, Yang\cite{46}, Lu et al.\cite{47}, Zhai et al.\cite{48}, Guzman et al.\cite{52},  Henke\cite{49}, Rashad et al.\cite{53}, Vasupongayya et al.\cite{54}, Sheshadri et al.\cite{55}, Younis et al.\cite{50}, is provided in Table \ref{table:2} by considering the following features to compare the proposed BSSSQS with the existing examination management system models. 


\begin{itemize}
  \item Secure login - This feature covers security in the login process like password encryption, random password, etc. Proposed BSSSQS along with all of the existing models which we mentioned in Table \ref{table:2} ensure secure login for the users.   
  \item Random QSP generation - This feature generates QSP randomly from a list of questions. So that, the person who provides questions cannot predict which question may come in the exam. BSSSQS randomly generates QSPs from the provided questions and among the existing systems, Chang et al.\cite{44}, Lu et al.\cite{47}, Zhai et al.\cite{48},  Kaya et. al\cite{45}, and Younis et al.\cite{50} applied this feature in their proposed models.
   
  \item QSP encryption - This feature encrypts QSP so that unauthorized persons can not access it.  QSPs can be encrypted employing symmetric encryptions or asymmetric encryptions. BSSSQS provides two-phase encryption process so that QSPs can be more secure and among the existing systems no one mentioned about encrypting QSPs. 
  \item Random QSP selection - This feature supports the selection of a QSP randomly before the exam.  The benefits of the random selection of QSP is that no one can guess the selected QSP. BSSSQS randomly select a QSP from the provided QSPs and among the existing systems, Henke\cite{49}, Kaya et. al\cite{45}, and Younis et al.\cite{50} supports this feature in their proposed models.
  \item Timestamp lock - This feature helps to impose a restriction of time on QSPs so that no one can access QSPs before the allowed time. BSSSQS imposes a timestamp lock on the QSPs along with an alert system. No one can access QSPs before the notified time and if anyone tries to access QSPs before the notification, an alert system notifies regarding this unauthorized activity and among of the existing models, only BSSSQS supports this concept.
  \item Blockchain support - This feature supports the inclusion of blockchain technology. In BSSSQS, QSPs store in the blockchain along with a smart contract.  BSSSQS utilizes the feature of data security and integrity of blockchain in order to maintain QSPs security throughout the system. Among of the mentioned models, only BSSSQS employs blockchain in the system.
  \item Distributed Sharing - This feature covers sharing of questions among the exam centre in a distributed way. BSSSQS shares QSPs among the remote exam centres with the support of blockchain. Due to blockchain, it becomes easy to maintain the security of QSPs remotely. However, apart from BSSSQS, others are fully centralized systems whose manages QSP and exam within itself.  
\end{itemize}

\begin{table*}
\caption{Performance comparison between BSSSQS and existing models}
\label{table:2}	
 \begin{tabular}{|>{\centering}p{2.7cm}|>{\centering}p{0.3cm} | >{\centering}p{0.3cm} | >{\centering}p{0.3cm} |  >{\centering}p{0.3cm} |  >{\centering}p{0.3cm} |  >{\centering}p{0.3cm} |  >{\centering}p{0.3cm} |  >{\centering}p{0.3cm} |  >{\centering}p{0.3cm} |  >{\centering}p{0.3cm} |  >{\centering}p{0.3cm} | p{0.3cm} |} 
 \hline
\rotatebox{45}{Features / Schemes}  & \rotatebox{90}{Yang\cite{46}} & \rotatebox{90}{Guzman\cite{52}} & \rotatebox{90}{Rashad\cite{53}} & \rotatebox{90}{Vasupongayya\cite{54} } & \rotatebox{90}{Sheshadri\cite{55}} & \rotatebox{90}{Chang\cite{44}} &   \rotatebox{90}{Lu\cite{47}} & \rotatebox{90}{Zhai\cite{48}}  &  \rotatebox{90}{Henke\cite{49}} &  \rotatebox{90}{Kaya\cite{45}} &  \rotatebox{90}{Younis\cite{50}}  & \rotatebox{90}{BSSSQS}\\ [0.5ex]
 \hline
 Secure login & yes & yes & yes & yes & yes & yes & yes & yes & yes & yes & yes & yes \\ 
 \hline 
 Random QSP generation & no & no & no  & no & no & yes & yes & yes & no & yes & yes & yes \\
 \hline 
 QSP encryption  & no & no & no  & no & no & no & no & no & no  & no&no & yes \\
 \hline 
 Random QSP selection & no & no & no  & no & no & no & no & no & yes & yes & yes& yes \\
 \hline 
 Timestamp lock  & no & no & no  & no & no & no & no & no & no & no & no & yes \\ 
 \hline 
 Blockchain support  & no & no & no  & no & no & no & no & no & no & no & no & yes \\ 
  \hline
 Distributed Sharing  & no & no & no  & no & no & no & no & no & no  & no & no& yes \\  [1ex] 
 \hline
 \end{tabular}
\end{table*}

\section{Conclusion}
	\label{sec:5}
	In this paper, we have proposed a secured QS scheme exploiting the security mechanism of blockchain. In this scheme, QSP experiences two-phase encryption in order to prevent unethical access before the exam. Moreover, a restriction of time is issued in the proposed scheme so that every minion has to wait for system permission in order to initiate the decryption process of QSP. Furthermore, QSP is selected by master employing the proposed randomize algorithm. The combination of these features can provide a secured QS system. We have analyzed BSSSQS's security by proposing different propositions and proved each proposition with respect to BSSSQS which demonstrate the feasibility of BSSSQS's security against different attacks. We have compared the performance of our proposed scheme with other existing education management techniques. Based on the theoretical comparison, it can be demonstrated that BSSSQS is more secure than other models. Due to the use of blockchain concept, unethical access to the proposed system will be most challenging. Moreover, we have designed BSSSQS in a flexible way, where the modules are loosely coupled; as a result, any module can be replaced with a new module easily. Therefore, it can be said that BSSSQS can be a promising approach for providing proper security to mitigate QPL problem in the future smart education system. Furthermore, the proposed system can also be used to share sensitive documents with little or no modification, which can be subjected to future works.

\section*{Acknowledgments}
This work was supported by the Global Excellent Technology Innovation Program (10063078) funded by the Ministry of Trade, Industry and Energy (MOTIE) of Korea.


\bibliography{BSSSQS}

\end{document}